\documentclass[a4paper,12pt]{article}
\usepackage[utf8]{inputenc}

\usepackage{test,graphicx,subcaption}
\usepackage[colorlinks,citecolor=blue,urlcolor=magenta]{hyperref}

\usepackage[style=authoryear-comp, 
sorting=nyt, 
dashed=false, 
maxcitenames=2, 
uniquelist=false,
uniquename=false,
giveninits=true, 
natbib, 
date=year 
]{biblatex}
\DeclareFieldFormat{pages}{#1} 
\renewbibmacro{in:}{\ifentrytype{article}{}{\printtext{\bibstring{in}\intitlepunct}}} 

\usepackage[inline,shortlabels]{enumitem}
\setlist[enumerate,1]{label=(\roman*)}

\usepackage[margin = 1.25in]{geometry}
\usepackage{setspace}
\title{Incentivizing Hidden Types in Secretary Problem}
\author{Longjian Li\thanks{Department of Economics, New York University. Email: \href{mailto:ll4830@nyu.edu}{ll4830@nyu.edu}.} \and Alexis Akira Toda\thanks{Department of Economics, Emory University. Email: \href{mailto:alexis.akira.toda@emory.edu}{alexis.akira.toda@emory.edu}.}}

\numberwithin{equation}{section}
\numberwithin{lem}{section}
\numberwithin{prop}{section}

\begin{document}

\maketitle

\begin{abstract}
We study a game between $N$ job applicants who incur a cost $c\in [0,1)$ (relative to the job value) to reveal their type during interviews and an administrator who seeks to maximize the probability of hiring the best applicant. We define a full learning equilibrium and prove its existence, uniqueness, and optimality. In equilibrium, the administrator accepts the current best applicant $n$ with probability $c$ if $n<n^*$ and with probability 1 if $n\ge n^*$ for a threshold $n^*$ independent of $c$. In contrast to the case without cost, where the success probability converges to $1/\e\approx 0.37$ as $N$ tends to infinity, with cost the success probability decays like $N^{-c}$.

\medskip



\end{abstract}

\section{Introduction}

How can we identify the next superstar? Solving this problem may be of interest to a film director who seeks to identify an actor that perfectly fits the role, a department that seeks to hire the best junior candidate, or a racket manufacturer that seeks to sponsor the Rafael Nadal of the next generation. In mathematics, this problem is known as the \emph{secretary problem}, which can be described as follows. A known number $N$ of job applicants are to be interviewed one by one in random order, all $N!$ possible orders being equally likely. The administrator is able at any time to rank the applicants that have so far been interviewed from the best to worst. As each applicant is interviewed, the administrator must either accept, in which case the search is terminated, or reject, in which case the next applicant is interviewed and the administrator faces the same choice as before. The administrator's objective is to maximize the probability of accepting the very best among all $N$ applicants. This description of the secretary problem closely follows \citet{Freeman1983}, who reviews the literature on various extensions. \citet{Ferguson1989} provides a brief solution, some historical background, and a literature review. 

The solution to the secretary problem is to reject the first $n^*-1$ applicants and accept the next applicant who is best among those already interviewed, where
\begin{equation}
n^*=\min\set{n:\sum_{k=n}^{N-1}\frac{1}{k}\le 1}.\label{eq:nstar}
\end{equation}
One can show that as the number of applicants $N$ tends to infinity, $n^*/N$ converges to $1/\e\approx 0.37$ as does the probability of hiring the best applicant. 
One limitation of the classical secretary problem is that it ignores the incentive of applicants to show up for interviews: because the first $n^*-1$ applicants are always rejected in the classical setting, they have no incentive to show up. But without interviewing applicants, the administrator cannot learn their ability.

In this paper, we extend the classical secretary problem to a setting where job applicants incur a cost for completing the interview and optimally decide whether to complete the interview or not. For instance, we can imagine a situation where the job application consists of two stages, where the first stage incurs a negligible cost (\eg, submitting the application package for being considered) but the second stage incurs a cost $c\in [0,1)$ relative to the job value (\eg, showing up for interview and performing some tasks). Another example might be that in musical competition shows, singers perform sequentially and can observe the performances of their competitors. Additionally, they may have access to real-time information on audience voting ranks or live popularity rankings on social media, allowing them to estimate their potential ranks. In our model, the administrator sequentially invites applicants for interviews, informs them of the performance of past applicants who completed the interviews, and promises (with commitment) the probability of acceptance conditional on the interview performance. Given the past performance and the probability of acceptance, the current applicant then decides whether to complete the interview or not. If the applicant completes the interview, their performance reveals their ability. The key difference of our model from the classical setting is that the applicants must be incentivized to reveal their type (\ie, complete the interview).

In this strategic setting, an equilibrium consists of a contingent plan of the administrator to accept the applicant (the probability of offering the job conditional on current and past performances) and contingent plans of applicants to complete the interview. We say that an equilibrium is \emph{full learning} if the administrator can always identify the current best applicant. We prove that a unique full learning equilibrium exists and characterize it using dynamic programming. Despite the presence of incentive problems, the equilibrium strategies turn out to be remarkably similar to the classical secretary problem. For job applicants, the optimal behavior is to complete the interview if and only if they are the best among those already interviewed. The equilibrium strategy of the administrator is a threshold strategy: the current best applicant $n$ is accepted with probability $c$ if $n<n^*$ and with probability 1 if $n\ge n^*$, where the threshold $n^*$ is given by \eqref{eq:nstar}, identical to the classical case without cost. We also prove that among all equilibria, the full learning equilibrium is optimal in the sense that it achieves the highest probability of hiring the best applicant.

Despite some similarities to the classical secretary problem, the presence of interview cost has a significant welfare implication. In the classical case without cost, the administrator is able to hire the best applicant with success probability approaching $1/\e\approx 0.37$ as the number of applicants $N$ tends to infinity. In contrast, with interview cost $c$, the success probability exhibits a power law decay $N^{-c}$ and hence converges to zero as $N\to\infty$, although the convergence is slow.

\subsection{Related literature}

The literature on the secretary problem is vast and there are many variations. When the abilities of applicants are independently drawn from a known distribution and the administrator incurs a cost for interviewing, the solution is a threshold strategy that can be characterized using dynamic programming; see \citet{Sakaguchi1961}. \citet{buchbinder2014secretary} introduce a new linear programming technique to design secretary algorithms, using judiciously chosen variables and constraints. Our problem is different because the distribution of abilities is unknown (as in the classical secretary problem) and the applicants incur an interview cost and hence must be incentivized to complete the interviews. 

A few papers consider the secretary problem with sampling costs. \citet{Lorenzen1979,Lorenzen1981} studies the secretary problem with a cumulative interview cost function. The optimal stopping rule is given by a more general class called island rules, instead of the traditional cutoff point rules. With interview cost, the optimal stopping rule will stop sooner. In our model, the applicants incur an interview cost, not the administrator. This represents an indirect opportunity cost for the administrator, who does not need to make a monetary transfer but has to sacrifice future winning probability to incentive applicants, in contrast to a direct sampling cost. Our result is different because the optimal stopping rule is still the cutoff point rule and the stopping speed does not change. In \citet{Samuels1985}, the administrator incurs a linear interview cost and a cumulative recall cost. In \citet{FerensteinEnns1988}, sampling is free but the recall
cost is levied each time an observation is held. Our problem is different because the applicants cannot be recalled, as in the classical setting.

A few papers consider the secretary problem in strategic settings. In \citet{Fushimi1981}, two employers compete for hiring the best applicant. \citet{AlpernGal2009} and \citet{AlpernGalSolan2010} study the implication of veto power of committee members when a selection committee interviews the applicants and the firm optimally designs the selection committee. \citet{HahnHoeferSmorodinsky2022} study the strategic interaction between an agent (sender) that evaluates the applicants and another (receiver) that makes the hiring decision and characterize the sender-optimal persuasion mechanism. Our problem is different because none of these papers consider the incentives of the applicants. The secretary problem can be generalized to online auctions. Incentive issues in online algorithms are considered by \citet{lavi2000competitive}, \citet{awerbuch2003reducing}, and \citet{hajiaghayi2004adaptive}. \citet{ErikssonSjostrandStrimling2007} study the asymptotic behavior in a model with two-sided matching with $N$ agents in each group (\eg, men and women searching for a mate) when agents seek to minimize the expected rank of the match. \citet{ChudjakowRiedel2013} solve the secretary problem under ambiguity aversion: specifically, the administrator's subjective probability that applicant $n$ is the current best is not the point mass $1/n$ as in the classical case but belongs to some interval $[a_n,b_n]$ and the administrator applies the minimax principle to make decisions. \citet{laffont1993theory}, \citet{bolton2004contract}, and \citet{kreps2023microeconomic} provide excellent materials on incentives in general.
\section{Model}

This section describes the model. The setup is similar to the classical secretary problem described in the introduction except the presence of incentive problems of applicants.

\paragraph{Agents}
There are $N+1$ agents, the administrator and $N\ge 2$ job applicants indexed by $n=1,\dots,N$. Applicant $n$ has ability $\theta_n>0$. Job applicants are sequentially invited for interviews.

\paragraph{Actions}

When invited for an interview, applicant $n$ chooses action $a_n\in \set{0,1}$, where $a_n=0$ corresponds to declining the interview and $a_n=1$ corresponds to completing the interview. The interview reveals the output $y_n=a_n\theta_n$ to the administrator.\footnote{As our argument uses only ordinal properties, the output could take the more general form $y_n=a_ng(\theta_n)$, where $g:(0,\infty)\to(0,\infty)$ is a strictly increasing function.} Immediately after the interview $n$, the administrator must choose action $d_n \in \set{0,1}$ based only on the history of observed outputs $(y_1,\dots,y_n)$, where $d_n=1$ corresponds to accepting applicant $n$ and $d_n=0$ corresponds to rejecting. The game ends if the administrator accepts an applicant (or rejects all applicants). If applicant $n$ is rejected, the administrator invites applicant $n+1$ for an interview. Once rejected, an applicant cannot be recalled.

\paragraph{Payoffs}

For an applicant, being accepted for the job yields payoff normalized to 1 and completing the interview (choosing $a_n=1$) costs $c\in [0,1)$. The applicant $n$'s payoff is
\begin{equation*}
\begin{cases*}
    1-c & if $d_n=1$, $a_n=1$,\\
    -c & if $d_n=0$, $a_n=1$,\\
0 & otherwise.
\end{cases*}
\end{equation*}
If the administrator accepts an applicant with ability $\theta$, the payoff to the administrator is
\begin{equation*}
\begin{cases*}
1 & if $\theta = \max_{1\le n\le N}\theta_n$,\\
0 & otherwise.
\end{cases*}
\end{equation*}
That is, the administrator cares only about hiring the best applicant. All agents are risk-neutral.

\paragraph{Information}

The applicants' abilities $\set{\theta_n}_{n=1}^N$ are realized before the interview process begins but are private information and known only to each applicant. The administrator believes that the rank orders of $\set{\theta_n}_{n=1}^N$ have no ties and are equally likely with probability $1/N!$. When applicant $n$ is invited for interview, the administrator presents the past outputs $(y_1,\dots,y_{n-1})$ to the applicant. Applicant $n$ chooses action $a_n\in\set{0,1}$ and the output $y_n=a_n\theta_n$ is observed. The administrator then chooses whether to accept or reject applicant $n$. After the interview process ends and all decisions have been made, the abilities $\set{\theta_n}_{n=1}^N$ become public information and the payoffs are realized.

Presenting past performance is crucial for understanding how applicants decide whether to reveal their ability. It also benefits to the administrator, as without revealing any information about past interviews, the administrator would need to incentivize all applicants, whereas in our model, the administrator only needs to incentivize the high ability applicants. Presenting past performance acts as a screening process to distinguish high-ability applicants from low-ability applicants, ensuring that only high-ability applicants would show up for the interview.

\paragraph{Strategies}

Let $H_n=\R_+^n$ be the set of outputs of the first $n$ applicants, with the convention $H_0=\emptyset$.  The administrator's (mixed) strategy is a collection of functions $\sigma=\set{\sigma_n}_{n=1}^N$ with
\begin{equation*}
\sigma_n:H_n\to [0,1],
\end{equation*}
where $p_n=\sigma_n(y_1,\dots,y_n)$ is the probability that the administrator accepts applicant $n$ given the outputs $(y_1,\dots,y_n)$. The administrator's strategy space is denoted by $\Sigma$. The administrator has commitment power and therefore chooses $\sigma$ once and for all, which is common knowledge.

Applicant $n$'s strategy is a function
\begin{equation*}
s_n:\Sigma \times H_{n-1}\times (0,\infty)\to \set{0,1},
\end{equation*}
where $s_n(\sigma,y_1,\dots,y_{n-1},\theta)=1$ ($=0$) means applicant $n$ with ability $\theta$ completes (declines) the interview given the past outputs $(y_1,\dots,y_{n-1})$ and the administrator's strategy $\sigma$. A Nash equilibrium consists of a strategy profile $(\sigma^*,s_1^*,\dots,s_N^*)$ that is mutually a best response.

\section{Full learning equilibrium}\label{sec:full_learn}

In this section we construct an equilibrium in which learning occurs. We define a full learning equilibrium as follows.

\begin{defn}\label{defn:full_learning}
We say that $(\sigma^*,s_1^*,\dots,s_N^*)$ is a \emph{full learning equilibrium} if for any equilibrium path and $n$ until the game ends, we have
\begin{equation}
\max_{1\le k\le n}\theta_k=\max_{1\le k\le n}y_k.\label{eq:full_learn}
\end{equation}
\end{defn}

The reason why we call an equilibrium in Definition \ref{defn:full_learning} ``full learning'' is the following. As the administrator cares only about hiring the highest ability applicant, the decision of whether to accept or reject applicant $n$ depends only on whether $\theta_n=\max_{1\le k\le n}\theta_k$ ($n$ is the current best) or $\theta_n<\max_{1\le k\le n}\theta_k$ ($n$ is not the current best). However, the ability $\theta_n$ is not directly observable and must be inferred from the output $y_n$. The following lemma shows that condition \eqref{eq:full_learn} is sufficient for determining whether $\theta_n$ is the current maximum or not.

\begin{lem}\label{lem:thetamax}
Let $(\sigma^*,s_1^*,\dots,s_N^*)$ be a full learning equilibrium. Then $\theta_1=y_1$ and
\begin{equation*}
\theta_n\begin{cases*}
=\max_{1\le k\le n}\theta_k & if $y_n>\max_{1\le k\le n-1}y_k$,\\
<\max_{1\le k\le n}\theta_k & if $y_n\le \max_{1\le k\le n-1}y_k$
\end{cases*}
\end{equation*}
for $n\ge 2$ until the game ends.
\end{lem}

\begin{proof}
$\theta_1=y_1$ is trivial by setting $n=1$ in \eqref{eq:full_learn}.

Suppose $n\ge 2$ and $y_n>\max_{1\le k\le n-1}y_k$. Since $y_k\ge 0$ for all $k$, it follows that $a_n\theta_n=y_n>0$ and hence $a_n=1$ and $\theta_n=y_n$. Using \eqref{eq:full_learn}, we obtain $\theta_n=y_n>\max_{1\le k\le n-1}y_k=\max_{1\le k\le n-1}\theta_k$, so $\theta_n=\max_{1\le k\le n}\theta_k$.

Next suppose that $y_n\le \max_{1\le k\le n-1}y_k$. Using \eqref{eq:full_learn}, we obtain
\begin{equation*}
\theta_n\le \max_{1\le k\le n}\theta_k=\max_{1\le k\le n}y_k=\max\set{\max_{1\le k\le n-1}y_k,y_n}=\max_{1\le k\le n-1}y_k=\max_{1\le k\le n-1}\theta_k,
\end{equation*}
and the inequality is strict because there are no ties in $\theta$.
\end{proof}

Using Lemma \ref{lem:thetamax}, we can provide necessary conditions for full learning equilibrium strategies.

\begin{lem}\label{lem:eq_strategy}
Let $(\sigma^*,s_1^*,\dots,s_N^*)$ be a full learning equilibrium. Then
\begin{subequations}
\begin{align}
&\sigma_n^*(y_1,\dots,y_n) \begin{cases*}
\ge c & if $y_n>\max_{1\le k\le n-1}y_k$,\\
=0 & if $y_n\le \max_{1\le k\le n-1}y_k$,
\end{cases*}\label{eq:sigman}\\
&s_n^*(\sigma^*,y_1,\dots,y_{n-1},\theta)=\begin{cases*}
1 & if $\theta>\max_{1\le k\le n-1}y_k$ and $\sigma_n^*(y_1,\dots,y_{n-1},\theta) \ge c$,\\
0 & otherwise.
\end{cases*}\label{eq:sn}
\end{align}
\end{subequations}
\end{lem}

\begin{proof}
Let $p_n\coloneqq \sigma_n^*(y_1,\dots,y_n)$ be the acceptance probability of applicant $n$. We first show \eqref{eq:sigman}. Suppose $y_n\le \max_{1\le k\le n-1}y_k$. By Lemma \ref{lem:thetamax}, we have $\theta_n<\max_{1\le k\le n}\theta_k$, so applicant $n$ is never the best and $p_n=0$ is optimal. Suppose $y_n>\max_{1\le k\le n-1}y_k$. By Lemma \ref{lem:thetamax} and \eqref{eq:full_learn}, we have
\begin{equation*}
0<\theta_n=\max_{1\le k\le n}\theta_k=\max_{1\le k\le n}y_k=y_n,
\end{equation*}
so $a_n=1$. If applicant $n$ chooses $a_n=1$, the expected payoff is $p_n-c$. If applicant $n$ chooses $a_n=0$, then $y_n=0$ and the second case of \eqref{eq:sigman} (which is already proved) implies $\sigma_n^*(y_1,\dots,y_{n-1},0)=0$, so the expected payoff is 0. Therefore choosing $a_n=1$ is optimal if and only if $p_n-c\ge 0\iff p_n\ge c$, which is the first case of \eqref{eq:sigman}.

We next show \eqref{eq:sn}. If applicant $n$ chooses $a_n=0$, the expected payoff is $\sigma_n^*(y_1,\dots,y_{n-1},0)=0$ by the second case of \eqref{eq:sigman}. If applicant $n$ chooses $a_n=1$, using \eqref{eq:sigman}, the expected payoff is
\begin{equation*}
\begin{cases*}
\sigma_n^*(y_1,\dots,y_{n-1},\theta)-c & if $\theta>\max_{1\le k\le n-1}y_k$,\\
-c<0 & if $\theta\le \max_{1\le k\le n-1}y_k$.
\end{cases*}
\end{equation*}
Comparing the expected payoffs from $a_n=0$ and $a_n=1$, we obtain the best response \eqref{eq:sn}.
\end{proof}

We now characterize the full learning equilibrium using a dynamic programming technique similar to \citet{Beckmann1990} in the classical setting. By Lemma \ref{lem:thetamax}, under full learning, the history of outputs $(y_1,\dots,y_n)$ is sufficient to determine whether the current applicant is the best among those already interviewed (\ie, $\theta_n=\max_{1\le k\le n}\theta_k$). Call this state ``1'', and call the state ``0'' otherwise (the current applicant is dominated by one already interviewed). Let $X=\set{0,1}$ be the state space and $V_n(x)$ be the expected payoff to the administrator (value function) in a full learning equilibrium when interviewing applicant $n$ in state $x\in X$.

Let us derive the Bellman equations. Suppose $2\le n\le N-1$ and $x_n=0$ (the current applicant is not the best among those already interviewed). Then by Lemma \ref{lem:eq_strategy} the administrator always rejects applicant $n$ and moves on to interviewing applicant $n+1$. Because the rank orders of $\set{\theta_1,\dots,\theta_N}$ are equally likely, so are the rank orders of $\set{\theta_1,\dots,\theta_{n+1}}$ and the conditional probability of the next state $x_{n+1}$ is
\begin{equation*}
\Pr(x_{n+1}=1 \mid x_n=0)=\Pr(x_{n+1}=1)=\Pr\left(\theta_{n+1}=\max_{1\le k\le n+1}\theta_k\right)=\frac{1}{n+1}.
\end{equation*}
Therefore we obtain the Bellman equation
\begin{equation}
V_n(0)=\frac{1}{n+1}V_{n+1}(1)+\frac{n}{n+1}V_{n+1}(0)\label{eq:bellman0}
\end{equation}
for $n\ge 2$. Note that since $\theta_1=\max_{1\le k\le 1}\theta_k$ trivially, the state is always $x_1=1$ when $n=1$, and therefore $V_1(0)$ is undefined. If we define $V_1(0)$ by \eqref{eq:bellman0} for $n=1$, then \eqref{eq:bellman0} holds for all $n=1,\dots,N-1$.

Suppose next that $1\le n\le N-1$ and $x_n=1$. By the definition of the state $x_n=1$, we have $\theta_n=\max_{1\le k\le n}\theta_k$. Lemma \ref{lem:thetamax} then implies $y_n>\max_{1\le k\le n-1}y_k$. Then $p_n\coloneqq \sigma_n^*(y_1,\dots,y_n)\ge c$ by \eqref{eq:sigman}. Conditional on $x_n=1$, $p_n\ge c$, and applicant $n$ is accepted, the expected payoff to the administrator is
\begin{align*}
    &\Pr\left(\theta_n=\max_{1\le k\le N}\theta_k \mid x_n=1\right)\\
    &=\Pr(\text{$n$ is best among all} \mid \text{$n$ is best among first $n$})\\
    &=\Pr(\text{$n$ is best among all and first $n$})/\Pr(\text{$n$ is best among first $n$})\\
    &=\Pr(\text{$n$ is best among all})/\Pr(\text{$n$ is best among first $n$})\\
    &=(1/N)/(1/n)=\frac{n}{N}.
\end{align*}

If applicant $n$ is rejected (with probability $1-p_n$), the continuation value is the same as in the case $x_n=0$. Therefore we obtain the Bellman equation
\begin{equation}
V_n(1)=\max_{c\le p_n\le 1}\set{p_n\frac{n}{N}+(1-p_n)\left(\frac{1}{n+1}V_{n+1}(1)+\frac{n}{n+1}V_{n+1}(0)\right)}.\label{eq:bellman1}
\end{equation}

The following proposition characterizes the value functions.

\begin{prop}\label{prop:value}
The value functions in a full learning equilibrium satisfy $V_N(0)=0$, $V_N(1)=1$, and
\begin{subequations}\label{eq:b}
\begin{align}
    V_n(0)&=\frac{1}{n+1}V_{n+1}(1)+\frac{n}{n+1}V_{n+1}(0),\label{eq:b0}\\
    V_n(1)&=\max_{c\le p_n\le 1}\set{p_n\frac{n}{N}+(1-p_n)V_n(0)}\label{eq:b1}\\
    &=\max\set{c\frac{n}{N}+(1-c)V_n(0),\frac{n}{N}}>0.\label{eq:Vn1max}
\end{align}
\end{subequations}
\end{prop}

\begin{proof}
If $x_N=0$ when interviewing the last ($N$-th) applicant, it means the last applicant is not the best and the administrator has failed to hire the best. Therefore $V_N(0)=0$. If $x_N=1$, the last applicant is the best and hence accepting with probability 1 is optimal. Therefore $V_N(1)=1$.

Equation \eqref{eq:b0} follows from \eqref{eq:bellman0} and the definition of $V_1(0)$.

Using \eqref{eq:bellman0} and \eqref{eq:bellman1}, we obtain \eqref{eq:b1}. Then \eqref{eq:Vn1max} holds because the objective function in \eqref{eq:b1} is linear in $p_n$ and the maximum occurs at one of the endpoints $p_n=c$ or $p_n=1$.
\end{proof}

Because the objective function in \eqref{eq:b1} is linear in $p_n$, it is clear that the optimal $p_n$ is $p_n=1$ if $V_n(0)\le n/N$ and $p_n=c$ if $V_n(0)>n/N$. Thus it remains to characterize when $V_n(0)/n\gtrless 1/N$. To this end, it is convenient to define the normalized value function $v_n(x)\coloneqq V_n(x)/n$. The following proposition derives a recursive formula for $v_n(x)$ and shows that it is decreasing in $n$.

\begin{prop}\label{prop:v}
The normalized value $v_n(x)=V_n(x)/n$ satisfies $v_N(0)=0$, $v_N(1)=1/N$, and
\begin{subequations}\label{eq:v}
\begin{align}
v_n(0)&=\frac{1}{n}v_{n+1}(1)+v_{n+1}(0),\label{eq:v0}\\
v_n(1)&=\max\set{c/N+(1-c)v_n(0),1/N}.\label{eq:v1}
\end{align}
\end{subequations}
Furthermore, $v_n(0)$ is strictly decreasing in $n$ and $v_n(1)$ is decreasing in $n$.
\end{prop}

\begin{proof}
$v_N(0)=0$ and $v_N(1)=1/N$ follow from Proposition \ref{prop:value}. \eqref{eq:v0} and \eqref{eq:v1} follow by dividing \eqref{eq:b0} and \eqref{eq:Vn1max} by $n$, respectively.

To show that $v_n(0)$ is strictly decreasing in $n$, note that \eqref{eq:v1} implies $v_n(1)\ge 1/N>0$ for all $n$. Then \eqref{eq:v0} implies $v_{n+1}(0)-v_n(0)=-v_{n+1}(1)/n<0$, so $v_n(0)$ is strictly decreasing in $n$. Since $v_n(1)$ depends on $n$ only through $v_n(0)$ as in \eqref{eq:v1}, it follows that $v_n(1)$ is decreasing in $n$.
\end{proof}

We can now characterize the administrator's equilibrium strategy. Since by Proposition \ref{prop:v} $v_n(0)=V_n(0)/n$ is strictly decreasing in $n$, there exists a threshold $n^*$ such that the administrator accepts the current best applicant $n<n^*$ ($n\ge n^*$) with probability $p_n=c$ ($p_n=1$). Surprisingly, the threshold $n^*$ depends only on the number of applicants $N$ and not on the cost $c$.

In what follows, we defer long proofs to Appendix \ref{sec:proof}.

\begin{prop}\label{prop:sigma}
Define $n^*$ by
\begin{equation}
n^*=\min\set{n:\sum_{k=n}^{N-1}\frac{1}{k}\le 1}.\label{eq:n*}
\end{equation}
Then in a full learning equilibrium, the administrator's strategy is
\begin{equation}
\sigma_n^*(y_1,\dots,y_n)=\begin{cases*}
1 & if $n\ge n^*$ and $0<y_n=\max_{1\le k\le n}y_k$,\\
c & if $n<n^*$ and $0<y_n=\max_{1\le k\le n}y_k$,\\
0 & otherwise.
\end{cases*}\label{eq:sigma}
\end{equation}
\end{prop}

Intuitively, the administrator with positive `learning cost' lacks incentive to interview more applicants than a zero-learning-cost administrator. On the other hand, if the administrator assigns probability one to applicant $n^*-1$, since the zero-learning-cost administrator shares the same knowledge, the zero-learning-cost administrator also should choose applicant $n^*-1$ with probability one, which contradicts the optimality of the threshold $n^*$.

We can now construct a full learning equilibrium.

\begin{thm}\label{thm:eq}
For all $N\ge 2$ and $c\in [0,1)$, there exists a unique full learning equilibrium, which can be constructed as follows:
\begin{enumerate}
\item Define $n^*\in \set{1,\dots,N}$ by \eqref{eq:n*}.
\item Define $\sigma_n^*:H_n\to [0,1]$ by \eqref{eq:sigma}.
\item Define $s_n^*:\Sigma \times H_{n-1}\times (0,\infty)\to \set{0,1}$ by \eqref{eq:sn}.
\end{enumerate}
Let $\sigma^*=\set{\sigma_n^*}_{n=1}^N$. Then $(\sigma^*,s_1^*,\dots,s_N^*)$ is a full learning equilibrium and the success probability is $\pi_N^*=v_1(1)$, which can be computed by iterating \eqref{eq:v}.
\end{thm}

Theorem \ref{thm:eq} shows that the full learning equilibrium can be described in words as follows. First, given the number of applicants $N$, define the threshold $n^*$ by \eqref{eq:n*}. Second, the administrator accepts the current best applicant $n$ with probability $c$ if $n<n^*$ and with probability 1 if $n\ge n^*$. Finally, applicant $n$ completes the interview (choose $a_n=1$) if and only if their ability exceeds past outputs $\set{y_1,\dots,y_{n-1}}$.

\section{Optimality of full learning}

In Section \ref{sec:full_learn}, we proved that a unique full learning equilibrium exists and constructed it using dynamic programming. However, the model has many other equilibria. For instance, consider the strategy profile $(\sigma^*,s_1^*,\dots,s_N^*)$ defined by $s_n^*\equiv 0$ for all $n$ and $\sigma_n^*\equiv p_n$, where $(p_1,\dots,p_N)$ is any fixed probability vector. Let us show that this is an equilibrium. First, given that applicant $n$ is accepted with fixed probability $p_n$ regardless of the output, there is no incentive to complete the interview (which is costly) and $s_n^*(\sigma^*,y_1,\dots,y_{n-1},\theta)=0$ is optimal for any $(y_1,\dots,y_{n-1},\theta)$. Second, given that applicants decline interviews under $(s_1^*,\dots,s_N^*)$, the administrator always observes the output $y_n=0$ and gains no information. Thus randomly accepting an applicant yields the success probability
\begin{equation*}
\sum_{n=1}^Np_n\frac{(N-1)!}{N!}=\frac{1}{N},
\end{equation*}
which is independent of $(p_1,\dots,p_N)$ and hence $(p_1,\dots,p_N)$ is optimal.

Clearly, there are many other equilibria with partial learning (\eg, ignore the first applicant and adhere to full learning for $n\ge 2$). Given that many equilibria exist, a natural question is which one achieves the highest success probability. The following theorem shows the optimality of the full learning equilibrium.

\begin{thm}\label{thm:optim}
The full learning equilibrium is optimal in the sense that the success probability $\pi_N^*$ is the highest among all equilibria.
\end{thm}
Theorem \ref{thm:optim} depends crucially on the uniform cost parameter assumption. In the case of heterogeneous cost, this theorem may fail. For example, consider the cost profile $(c_1,\dots,c_N)=(1,0,0,\dots,0)$, where $c_n$ denotes the cost of applicant $n$. Because $c_1=1$, the only way to incentivize applicant 1 is to accept with probability 1, resulting in success probability $1/N$. Clearly, the administrator should ignore the first applicant. To analyze the heterogeneous cost problem formally, we need to assume that the administrator knows all applicants' costs ex-ante or at least knows the cost distribution if costs are drawn independently. The first scenario, where the administrator knows all applicants' costs, is a very strong assumption and not realistic. In the second scenario, where the administrator knows the cost distribution, the administrator's strategy might become path-dependent and hard to specify.

\section{Asymptotic behavior}

Consider the secretary problem with $N\ge 2$ applicants. Let $n_N^*$ be the threshold for experimenting determined by \eqref{eq:n*} and $\pi_N^*$ be the success probability under full learning. In this section we characterize the asymptotic behaviors of $n_N^*$ and $\pi_N^*$ as $N\to\infty$.

\begin{prop}\label{prop:nstar}
The threshold $n_N^*$ in \eqref{eq:n*} is increasing in $N$ and satisfies
\begin{equation}
\frac{N}{\e}\le n_N^*\le \frac{N-1}{\e}+2. \label{eq:n*bd}
\end{equation}
In particular, $\lim_{N\to\infty} n_N^*/N=1/\e=0.367\dotsc$.
\end{prop}

The success probability $\pi_N^*$ can be explicitly computed; see \eqref{eq:piN} in the Appendix. The following theorem shows that $\pi_N^*$ exhibits a power law decay with exponent $-c$.

\begin{thm}\label{thm:pilim}
Let $\Gamma$ be the gamma function. If $c\in [0,1)$, then
\begin{equation}
\lim_{N\to\infty}N^c\pi_N^*=\frac{\e^{c-1}}{\Gamma(2-c)}.\label{eq:pi*c>0}
\end{equation}
In particular, if $c=0$ then $\lim_{N\to\infty} \pi_N^*=1/\e=0.367\dotsc$.
\end{thm}

The administrator may be interested not only in the success probability but also in the expected length of search. To address this issue, let $\tau_N^*$ be the index of the accepted applicant in the full learning equilibrium, which is a random variable. The following theorem shows that $\E[\tau_N^*]$ equals $N\pi_N^*$ and consequently exhibits a power law with exponent $1-c$.

\begin{thm}\label{thm:tau}
If $c\in [0,1)$, then $\E[\tau_N^*]=N\pi_N^*$ and
\begin{equation}
\lim_{N\to\infty} N^{c-1}\E[\tau_N^*]=\frac{\e^{c-1}}{\Gamma(2-c)}.\label{eq:tau}
\end{equation}
\end{thm}

As an illustration, we consider a numerical example. Suppose we set the cost to $c=0.1$. 
For each $N$, we can easily compute $\pi_N^*=v_1(1)$ by iterating \eqref{eq:v}. Figure \ref{fig:prob} shows $\pi_N^*$ for $1\le N\le 1{,}000$, both for $c=0$ and $c=0.1$. When $c=0$, by Theorem \ref{thm:pilim}, the success probability $\pi_N^*$ converges to $1/\e\approx 0.37$ as $N\to\infty$. When $c=0.1$, by Theorem \ref{thm:pilim} $\pi_N^*$ converges to 0 but the convergence is slow. In fact, even with $N=1{,}000$ applicants, $\pi_N^*$ is over 0.2.

Figure \ref{fig:prob_loglog} shows the success probability $\pi_N^*$ for $1\le N\le 10^6$ in a log-log scale when $c=0.1$ as well as the asymptotic approximation
\begin{equation*}
\pi_N^*\sim \frac{\e^{c-1}}{\Gamma(2-c)}N^{-c}
\end{equation*}
implied by \eqref{eq:pi*c>0}. Since $\pi_N^*$ decays like $N^{-c}$, in a log-log scale we observe a straight line pattern with slope equal to $-c$.

\begin{figure}[!htb]
\centering
\begin{subfigure}{0.48\linewidth}
\includegraphics[width=\linewidth]{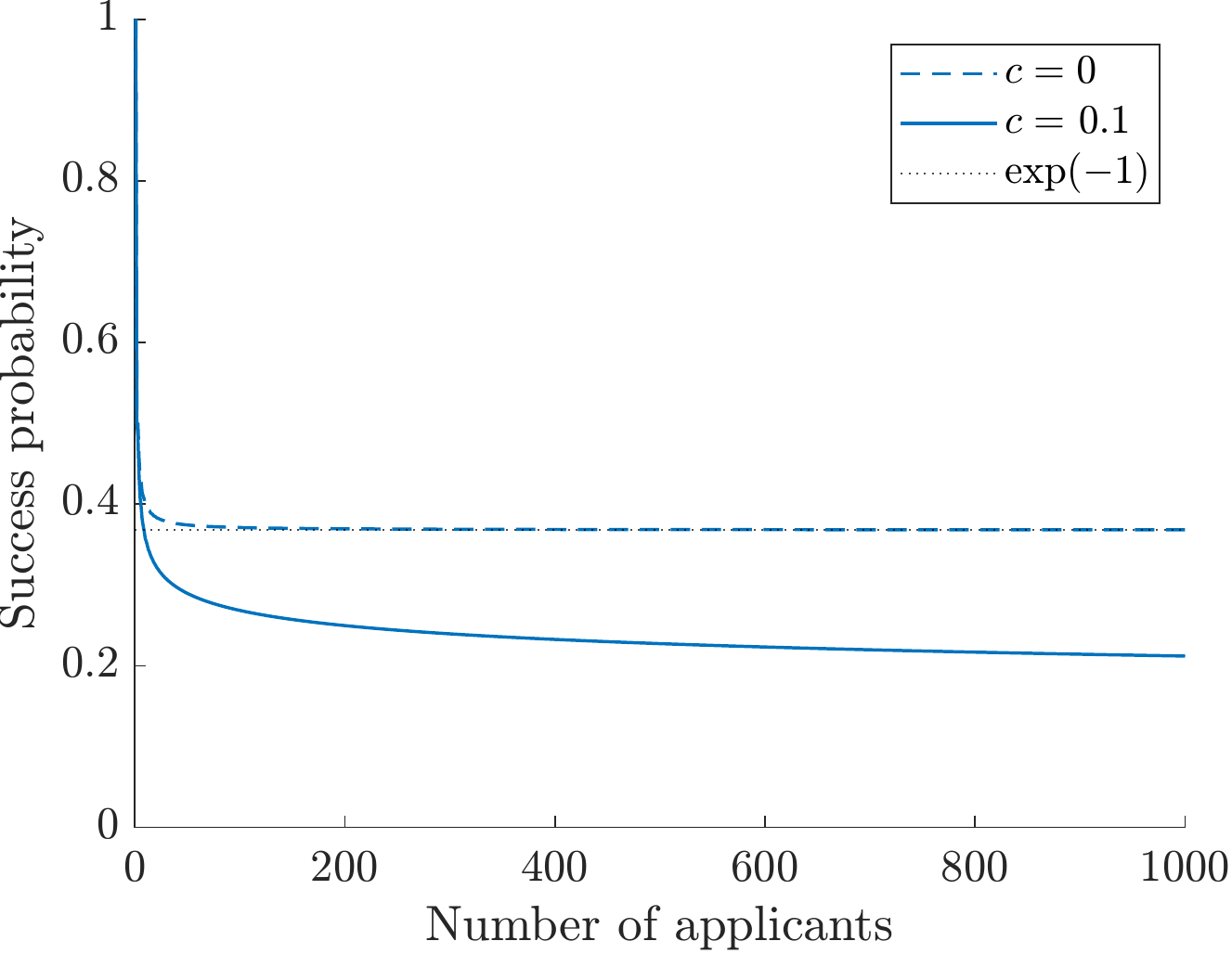}
\caption{Success probability with cost.}\label{fig:prob}
\end{subfigure}
\begin{subfigure}{0.48\linewidth}
\includegraphics[width=\linewidth]{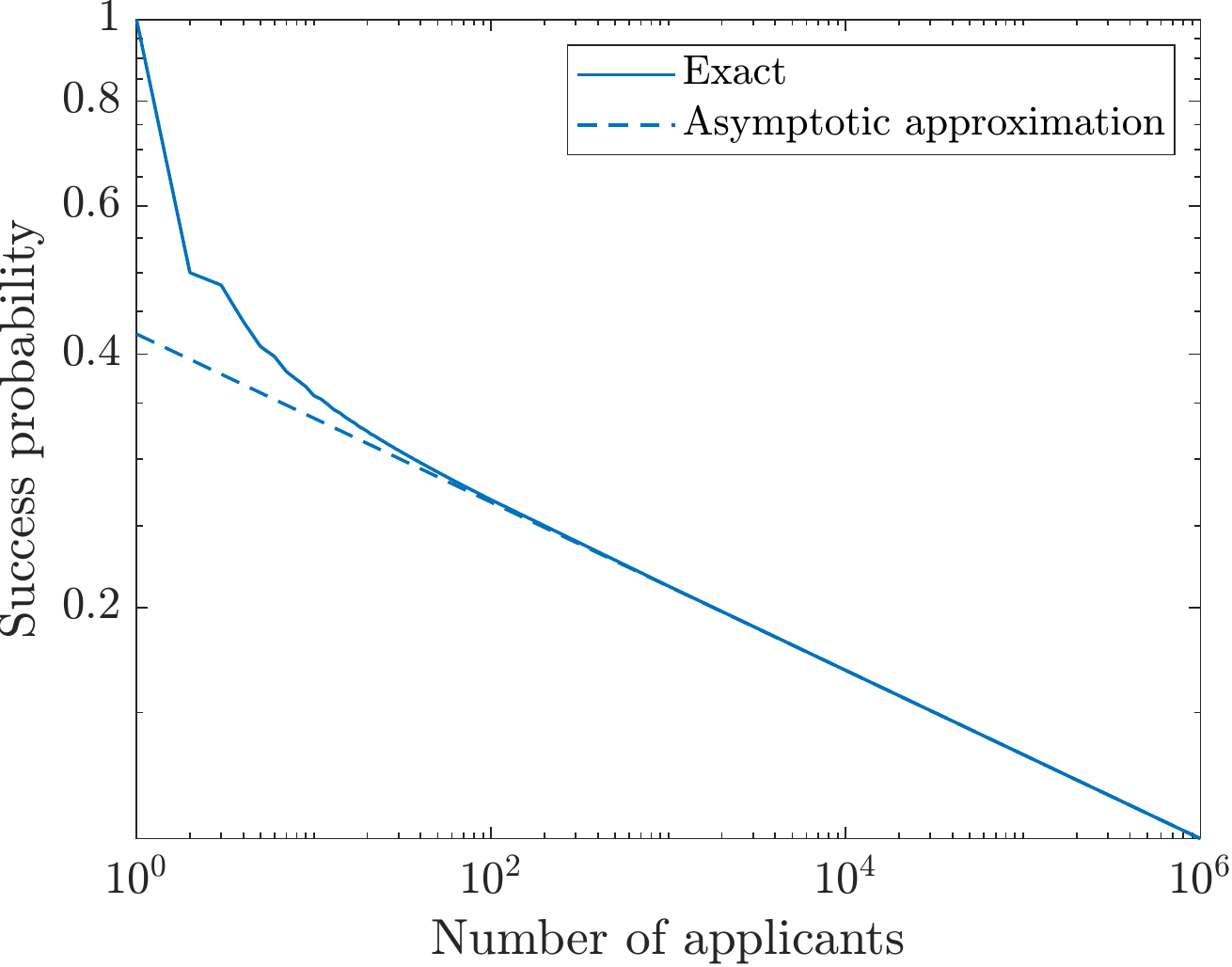}
\caption{Log-log plot of success probability.}\label{fig:prob_loglog}
\end{subfigure}
\end{figure}

\appendix

\section{Proofs}\label{sec:proof}

\subsection{Proof of Proposition \ref{prop:sigma}}

By the remark after Proposition \ref{prop:value}, the optimal behavior of the administrator is to accept the current best applicant with probability $p_n=1$ if $v_n(0)=V_n(0)/n\le 1/N$ and with probability $p_n=c$ if $v_n(0)=V_n(0)/n>1/N$. By Proposition \ref{prop:v}, we have $v_N(0)=0<1/N$. Therefore the threshold
\begin{equation}
t\coloneqq \min\set{n:v_n(0)\le 1/N}\label{eq:t}
\end{equation}
is well defined. Since by Proposition \ref{prop:v} $v_n(0)$ is strictly decreasing in $n$, it follows from the definition of $t$ that $v_n(0)\le 1/N$ for $n\ge t$ and $v_n(0)>1/N$ for $n<t$.

Let us show that $t$ equals $n^*$ in \eqref{eq:n*}. By the definition of $t$, we have $v_n(0)\le 1/N$ if and only if $n\ge t$. Then \eqref{eq:v1} implies
\begin{equation*}
v_{n+1}(1)=\max\set{c/N+(1-c)v_{n+1}(0),1/N}=1/N
\end{equation*}
for $n+1\ge t\iff n\ge t-1$, and \eqref{eq:v0} implies
\begin{equation}
v_n(0)=\frac{1}{Nn}+v_{n+1}(0)\label{eq:v0t}
\end{equation}
for $n\ge t-1$. Iterating \eqref{eq:v0t} and using $v_N(0)=0$, it follows that
\begin{equation}
v_n(0)=\frac{1}{N}\left(\frac{1}{n}+\frac{1}{n+1}+\dots+\frac{1}{N-1}\right)=\frac{1}{N}\sum_{k=n}^{N-1}\frac{1}{k}\label{eq:vn0}
\end{equation}
for $n\ge t-1$. In particular, setting $n=t$ in \eqref{eq:vn0} and using \eqref{eq:t}, we obtain
\begin{equation}
\frac{1}{N}\sum_{k=t}^{N-1}\frac{1}{k}=v_t(0)\le \frac{1}{N}\iff \sum_{k=t}^{N-1}\frac{1}{k}\le 1.\label{eq:tcond1}
\end{equation}
Comparing \eqref{eq:tcond1} to the definition of $n^*$ in \eqref{eq:n*}, we obtain $n^*\le t$. If $t=1$, then $1\le n^*\le t=1$, so $n^*=t$. Suppose $t\ge 2$. Since \eqref{eq:vn0} holds for $n\ge t-1$, setting $n=t-1$ and using the definition of $t$ in \eqref{eq:t}, we obtain
\begin{equation}
\frac{1}{N}\sum_{k=t-1}^{N-1}\frac{1}{k}=v_{t-1}(0)>\frac{1}{N}\iff \sum_{k=t-1}^{N-1}\frac{1}{k}>1.\label{eq:tcond2}
\end{equation}
Comparing \eqref{eq:tcond2} to the definition of $n^*$ in \eqref{eq:n*}, we obtain $n^*>t-1$. Therefore $n^*\ge t$ and hence $n^*=t$.

Finally, since $V_n(0)> n/N$ for $n<n^*$ and $V_n(0)\le n/N$ for $n\ge n^*$, it follows from \eqref{eq:b1} that accepting applicant $n$ with probability $p_n=\sigma_n^*(y_1,\dots,y_n)$ given by \eqref{eq:sigma} is optimal. \hfill \qedsymbol

\subsection{Proof of Theorem \ref{thm:eq}}
Uniqueness of $\sigma^*$ follows from Proposition \ref{prop:sigma}. Uniqueness of $s_n^*$ then follows from Lemma \ref{lem:eq_strategy}. The optimality of $s_n^*$ follows from Lemma \ref{lem:eq_strategy}. The optimality of $\sigma^*$ follows from Proposition \ref{prop:sigma}. Therefore $(\sigma^*,s_1^*,\dots,s_N^*)$ is an equilibrium.

To show that the equilibrium is full learning, it suffices to show \eqref{eq:full_learn}. Let us prove this by induction. For $n=1$, the condition $y_n>\max_{1\le k\le n-1}y_k$ trivially holds (because the maximum over an empty set is $-\infty$ by convention) and $s_1^*(\theta_1)\equiv 1$ by \eqref{eq:sn}. Therefore $y_1=a_1\theta_1=\theta_1$, so \eqref{eq:full_learn} holds. Suppose \eqref{eq:full_learn} holds for some $n$ and consider $n+1$. Let $\theta=\theta_{n+1}$.

If $\theta\le \max_{1\le k\le n}y_k$, by \eqref{eq:sn} we have $a_{n+1}=s_{n+1}^*(y_1,\dots,y_n,\theta)=0$, so $y_{n+1}=a_{n+1}\theta_{n+1}=0$. Since $\theta_{n+1}=\theta\le \max_{1\le k\le n}y_k=\max_{1\le k\le n}\theta_k$ by assumption and the induction hypothesis, we obtain
\begin{equation*}
\max_{1\le k\le n+1}y_k=\max\set{\max_{1\le k\le n}y_k,0}=\max_{1\le k\le n}y_k=\max_{1\le k\le n}\theta_k=\max_{1\le k\le n+1}\theta_k.
\end{equation*}

If $\theta>\max_{1\le k\le n}y_k$, by \eqref{eq:sn} we have $a_{n+1}=s_{n+1}^*(y_1,\dots,y_n,\theta)=1$, so $y_{n+1}=\theta_{n+1}$. Hence by assumption and the induction hypothesis, we obtain
\begin{equation*}
\max_{1\le k\le n+1}y_k=\max\set{\max_{1\le k\le n}y_k,\theta_{n+1}}=\max\set{\max_{1\le k\le n}\theta_k,\theta_{n+1}}=\max_{1\le k\le n+1}\theta_k.
\end{equation*}

Since \eqref{eq:full_learn} holds for $n+1$ regardless of $\theta$, by induction \eqref{eq:full_learn} holds for all $n$ until the game ends, and hence the equilibrium is full learning.

Finally, because the first applicant trivially satisfies $\theta_1=\max_{1\le k\le 1}\theta_k$, the state starts at $x_1=1$ and the probability of hiring the best applicant (success probability) is $\pi_N^*=V_1(1)=V_1(1)/1=v_1(1)$. \hfill \qedsymbol

\subsection{Proof of Theorem \ref{thm:optim}}

To prove Theorem \ref{thm:optim}, we establish a series of lemmas. Below, when necessary let $V_{n,N}(x)$ be the value function in the full learning equilibrium when the total number of applicants is $N$ and the administrator is interviewing applicant $n$ in state $x$.

\begin{lem}\label{lem:VnN0}
For $n=2,\dots,N$,
\begin{equation}
V_{n-1}(0)=\max_{c\le p_{n-1}\le 1}\set{\frac{p_{n-1}}{N}+\left(1-\frac{p_{n-1}}{n}\right)V_{n}(0)}\label{eq:VnN0}
\end{equation}
holds. Furthermore, if we define $V_{0}(0)$ by setting $n=1$ in \eqref{eq:VnN0}, we have $\pi_N^*=V_{0}(0)$.
\end{lem}

\begin{proof} 
Combining \eqref{eq:b0} and \eqref{eq:b1}, for $n\ge 2$ we obtain
\begin{align*}
V_{n-1}(0)&=\frac{1}{n}\max_{c\le p_{n-1}\le 1}\set{p_{n-1}\frac{n}{N}+(1-p_{n-1})V_{n}(0)}+\frac{n-1}{n}V_{n}(0)\\
&=\max_{c\le p_{n-1}\le 1}\set{\frac{p_{n-1}}{N}+\left(1-\frac{p_{n-1}}{n}\right)V_{n}(0)},
\end{align*}
which is \eqref{eq:VnN0}. Since $V_{0}(0)$ is defined by setting $n=1$ in \eqref{eq:VnN0}, we obtain
\begin{equation*}
V_{0}(0)=\max_{c\le p_{n-1}\le 1}\set{\frac{p_{n-1}}{N}+(1-p_{n-1})V_{1}(0)}=V_{1}(1)=\pi_N^*
\end{equation*}
by \eqref{eq:b1}.
\end{proof}

\begin{lem}\label{lem:VnNlb}
$V_{n,N}(1)\ge n/N$ and $V_{n,N}(0)\ge 1/N$ for $n<N$.
\end{lem}

\begin{proof}
$V_{n,N}(1)\ge n/N$ immediately follows from \eqref{eq:Vn1max}. Since $V_{n,N}(0)\ge 0$, setting $p_{n-1}=1$ in \eqref{eq:VnN0} yields $V_{n-1,N}(0)\ge 1/N$ for all $n$, and hence $V_{n,N}(0)\ge 1/N$ For all $n<N$.
\end{proof}

\begin{lem}\label{lem:NVnN}
For each $n$ and $x\in\set{0,1}$, $NV_{n,N}(x)$ is increasing in $N\ge n$.
\end{lem}

\begin{proof}
We first show that $NV_{n,N}(0)$ is increasing in $N$. Multiplying both sides of \eqref{eq:VnN0} by $N$ and letting $z_{n,N}\coloneqq NV_{n,N}(0)$, we obtain
\begin{equation}
z_{n-1,N}=\max_{c\le p_{n-1}\le 1}\set{p_{n-1}+(1-p_{n-1}/n)z_{n,N}}.\label{eq:bnN}
\end{equation}
For $n=1,2,\dotsc$, define the operator $T_n:[0,\infty)\to [0,\infty)$ by
\begin{equation*}
T_nz\coloneqq \max_{c\le p_{n-1}\le 1}\set{p_{n-1}+(1-p_{n-1}/n)z}.
\end{equation*}
Note that $z\mapsto T_nz$ is increasing in $z\ge 0$ because $p_{n-1}\le 1$ and $n\ge 1$. Since $V_{N,N}(0)=0$ and $V_{N-1,N}(0)=1/N$, we have $z_{N,N}=NV_{N,N}(0)=0$ and $z_{N-1,N}=NV_{N-1,N}(0)=1$. Therefore applying $T_{N-1},\dots,T_{n+1}$ (in this order) to $z_{N-1,N-1}=0<1=z_{N-1,N}$, we obtain
\begin{equation*}
z_{n,N-1}=T_{n+1}\dotsb T_{N-1}z_{N-1,N-1}\le T_{n+1}\dotsb T_{N-1}z_{N-1,N}=z_{n,N},
\end{equation*}
so $NV_{n,N}(0)=z_{n,N}$ is increasing in $N$.

Next we consider $NV_{n,N}(1)$. Multiplying both sides of \eqref{eq:b1} by $N$, we obtain
\begin{equation*}
NV_{n,N}(1)=\max_{c\le p_{n-1}\le 1}\set{p_{n-1}n+(1-p_{n-1})NV_{n,N}(0)},
\end{equation*}
which is increasing in $N$ because $NV_{n,N}(0)$ is.
\end{proof}

\begin{proof}[Proof of Theorem \ref{thm:optim}]
Take any equilibrium. Let $P_{0,N}=\pi_N$ be the success probability and $P_{n,N}(y_1,\dots,y_n)$ be the success probability conditional on having incentivized applicants up to $n$ to complete the interview (\ie, accepting with probability at least $c$ whenever the output achieves a new maximum and rejecting otherwise), observing outputs $y_1,\dots,y_n$, and adhering to the equilibrium strategy from that point on. Let $h_n=(y_1,\dots,y_n)$ be the history of outputs. 

Define $V_{0,N}(0)$ using \eqref{eq:b1}, which implies $V_{0,N}(0)=V_{1,N}(1)$. Then clearly the conclusions of Lemmas \ref{lem:VnNlb} and \ref{lem:NVnN} hold for $n=0$ and state $x=0$. Let us show by induction on $j=N-n$ that
\begin{equation}
P_{n,N}(h_n)\le \begin{cases*}
V_{n,N}(1) & if $y_n>\max_{1\le k\le n-1}y_k$,\\
V_{n,N}(0) & if $n=0$ or $y_n\le \max_{1\le k\le n-1}y_k$.
\end{cases*}\label{eq:Pnub}
\end{equation}
Since by assumption the administrator incentivizes the applicants up to $n$ to complete the interview, by Lemma \ref{lem:thetamax} we have
\begin{subequations}
\begin{align}
    y_n>\max_{1\le k\le n-1}y_k & \iff \theta_n>\max_{1\le k\le n-1}\theta_k, \label{eq:case1}\\
    y_n\le \max_{1\le k\le n-1}y_k & \iff \theta_n<\max_{1\le k\le n-1}\theta_k, \label{eq:case2}
\end{align}
\end{subequations}
where the strict inequality in \eqref{eq:case2} holds because there are no ties in $\theta$.

If $j=0$, then $n=N$. By assumption, full learning has occurred up to applicant $N$, so $P_{N,N}(h_N)=1$ if $y_N>\max_{1\le k\le N-1}y_k$ and $P_{N,N}(h_N)=0$ otherwise. Since $V_{N,N}(0)=0$ and $V_{N,N}(1)=1$, \eqref{eq:Pnub} holds.

Now suppose \eqref{eq:Pnub} holds whenever $N-n<j$ and consider $N-n=j<N$. Consider the second case in \eqref{eq:Pnub}, namely $n=0$ or $y_n\le \max_{1\le k\le n-1}y_k$. By the full learning assumption, applicant $n$ is not the best (if $n\ge 1$) or does not exist (if $n=0$) and in equilibrium the game continues. Consider arbitrary deviations from the equilibrium strategy. There can be two types of deviations: the first type is to decide not to learn the ability of applicant $n'=n+1$ but accept with arbitrary probability $p_n'\in [0,1]$, and the second type is to learn but accept with arbitrary probability $p_n'\in [c,1]$ conditional on being the current best.

Consider the first type of deviation. Then applicant $n'=n+1$ chooses $a_{n'}=0$. Since the probability that applicant $n'$ is the best is $1/N$, the continuation value is
\begin{equation*}
\underbrace{p_n'\frac{1}{N}}_{\text{$n'$ is best and accept}}+\underbrace{(1-p_n')\frac{N-1}{N}}_{\text{$n'$ is not best and reject}}\tilde{P}_{n,N-1}(h_n),
\end{equation*}
where $\tilde{P}$ is some equilibrium continuation value in a game with a total of $N-1$ applicants. (Note that since applicant $n'=n+1$ was not the best and rejected without learning, it is as if applicant $n'$ did not exist in the first place.) Since $p_n'\in [0,1]$ is arbitrary, by the induction hypothesis we obtain
\begin{equation*}
P_{n,N}(h_n)\le \max_{0\le p_n'\le 1}\set{p_n'\frac{1}{N}+(1-p_n')\frac{N-1}{N}V_{n,N-1}(0)}.
\end{equation*}
Since $V_{n,N-1}(0)\ge \frac{1}{N-1}$ by Lemma \ref{lem:VnNlb} (which is also valid for $n=0$ because $V_{0,N}(0)=V_{1,N}(1)$), the maximum is achieved when $p_n'=0$. Therefore
\begin{equation*}
P_{n,N}(h_n)\le \frac{N-1}{N}V_{n,N-1}(0)\le V_{n,N}(0)
\end{equation*}
by Lemma \ref{lem:NVnN}.

Next we consider the second type of deviation. By the same argument as the derivation of the Bellman equation \eqref{eq:bellman0} and the induction hypothesis, the continuation value can be bounded as
\begin{equation*}
P_{n,N}(h_n)\le \frac{1}{n+1}V_{n+1,N}(1)+\frac{n}{n+1}V_{n+1,N}(0)=V_{n,N}(0).
\end{equation*}
Therefore under either type of deviation, \eqref{eq:Pnub} holds for the second case when $N-n=j$.

Now consider the first case in \eqref{eq:Pnub}, namely $y_n>\max_{1\le k\le n-1}y_k$. Suppose the administrator accepts applicant $n$ with probability $p_n$. Since by assumption $y_n>0$, it must be $p_n\ge c$ to incentivize the applicant to complete the interview. In this case, by exactly the same argument as in the derivation of the Bellman equation \eqref{eq:bellman1} and the proof of the second case of \eqref{eq:Pnub}, we obtain the upper bound
\begin{equation*}
P_{n,N}(h_n)\le \max_{c\le p_n \le 1}\set{p_n\frac{n}{N}+(1-p_n)V_{n,N}(0)}=V_{n,N}(1).
\end{equation*}
Therefore \eqref{eq:Pnub} holds for the first case when $N-n=j$. By induction, \eqref{eq:Pnub} holds for all $n,N$.

Setting $n=0$ in \eqref{eq:Pnub} yields $P_{0,N}\le V_{0,N}(0)=V_{1,N}(1)=\pi_N^*.$
\end{proof}

\subsection{Proof of Proposition \ref{prop:nstar}}

For fixed $n$, since the sum $\sum_{k=n}^{N-1}1/k$ is increasing in $N$, we obtain
\begin{equation*}
\set{n:\sum_{k=n}^{N-1}\frac{1}{k}\le 1}\supset \set{n:\sum_{k=n}^N\frac{1}{k}\le 1}.
\end{equation*}
Taking the minimum of these sets and using \eqref{eq:n*}, we obtain $n_N^*\le n_{N+1}^*$, so $n_N^*$ is increasing in $N$.

Next we derive the bound \eqref{eq:n*bd}. To simplify the notation, let $t=n_N^*$. Using \eqref{eq:n*}, it follows that
\begin{equation*}
1\ge \sum_{k=t}^{N-1}\frac{1}{k}\ge \int_t^N\frac{1}{x}\diff x=\log\frac{N}{t} \implies t\ge \frac{N}{\e},
\end{equation*}
which is the lower bound in \eqref{eq:n*bd}.

Since $N\ge 2$, the upper bound \eqref{eq:n*bd} trivially holds if $t\le 2$, so assume $t\ge 3$. Then
\begin{equation*}
1<\sum_{k=t-1}^{N-1}\frac{1}{k}\le \int_{t-2}^{N-1}\frac{1}{x}\diff x= \log\frac{N-1}{t-2}\implies t\le \frac{N-1}{\e}+2,
\end{equation*}
which is the upper bound in \eqref{eq:n*bd}. \hfill \qedsymbol

\subsection{Proof of Theorem \ref{thm:pilim}}

To prove Theorem \ref{thm:pilim}, we need several lemmas.

\begin{lem}\label{lem:sumbound}
For $N\ge 6$, we have
\begin{equation}
    1<\sum_{n=n_N^*}^N\frac{1}{n-1}\le 1+\frac{1}{n_N^*-1}. \label{eq:sumbound}
\end{equation}
\end{lem}

\begin{proof}
By Proposition \ref{prop:nstar}, for $N\ge 6$ we have $n_N^*\ge 2$. To simplify the notation, let $t=n_N^*$. Using \eqref{eq:n*}, we obtain
\begin{equation*}
\sum_{k=t}^{N-1}\frac{1}{k}\le 1<\sum_{k=t-1}^{N-1}\frac{1}{k}.
\end{equation*}
Setting $k=n-1$ yields
\begin{equation*}
    \sum_{n=t}^N\frac{1}{n-1}-\frac{1}{t-1}=\sum_{n=t+1}^N\frac{1}{n-1}\le 1<\sum_{n=t}^N\frac{1}{n-1},
\end{equation*}
which is equivalent to \eqref{eq:sumbound}.
\end{proof}

\begin{lem}\label{lem:Sn}
Let $\Gamma$ be the gamma function. For $0\le c<1$, let $S_0(c)=1$ and
\begin{equation}
    S_n(c)\coloneqq \prod_{k=1}^n\left(1-\frac{c}{k}\right)\label{eq:Sn}
\end{equation}
for $n\ge 1$. Then
\begin{equation}
    \lim_{n\to\infty}n^cS_n(c)=\frac{1}{\Gamma(1-c)}.\label{eq:Snlim}
\end{equation}
\end{lem}

\begin{proof}
We start from the Gauss product formula of the gamma function
\begin{equation}
\Gamma(z)=\lim_{n\to\infty}\frac{n^z n!}{z(z+1)\dotsb(z+n)},\label{eq:gamma_gauss_prod}
\end{equation}
which is valid for all $z\in\C$ except nonpositive integers; see Equation (2.7) of \citet{Artin1964} or Section 8.19 of \citet{Rudin1976}. Multiplying both sides of \eqref{eq:gamma_gauss_prod} by $z$ and taking the reciprocal, we obtain
\begin{equation*}
\frac{1}{\Gamma(1+z)}=\frac{1}{z\Gamma(z)}=\lim_{n\to\infty}n^{-z}(1+z)\dotsb(1+z/n)=\lim_{n\to\infty}n^{-z}S_n(-z).
\end{equation*}
Letting $c=-z$ yields \eqref{eq:Snlim}.
\end{proof}

\begin{proof}[Proof of Theorem \ref{thm:pilim}]
By Proposition \ref{prop:nstar}, we have $n^*\to\infty$ as $N\to\infty$. Take $N\ge 6$ so that $n^*\ge 2$ by Proposition \ref{prop:nstar}. In the proof of Proposition \ref{prop:sigma}, we showed that the threshold $t$ defined by \eqref{eq:t} equals $n^*$ in \eqref{eq:n*}. Consequently, we have $V_n(0)/n\le 1/N$ for $n\ge n^*$ and $V_n(0)/n>1/N$ for $n<n^*$, and \eqref{eq:VnN0} reduces to
\begin{equation}
V_{n-1}(0)=\begin{cases*}
c/N+(1-c/n)V_n(0) & if $n<n^*$,\\
1/N+(1-1/n)V_n(0) & if $n\ge n^*$.
\end{cases*}\label{eq:Vn0}
\end{equation}
Iterating \eqref{eq:Vn0} for $n=1,\dots,n^*-1$ and using the definition of $S_n(c)$ in \eqref{eq:Sn}, we obtain
\begin{equation}
V_0(0)=\frac{c}{N}\sum_{n=1}^{n^*-1}S_{n-1}(c)+S_{n^*-1}(c)V_{n^*-1}(0).\label{eq:V0}
\end{equation}
Iterating \eqref{eq:Vn0} for $n=n^*,\dots,N$ and noting that $V_N(0)=0$, we obtain
\begin{equation}
V_{n^*-1}(0)=\frac{n^*-1}{N}\sum_{n=n^*}^N\frac{1}{n-1}.\label{eq:Vn*}
\end{equation}
Using Lemma \ref{lem:VnN0} and combining \eqref{eq:V0} and \eqref{eq:Vn*}, we obtain
\begin{equation}
N^c\pi_N^*=N^cV_0(0)=N^c\frac{c}{N}\sum_{n=1}^{n^*-1}S_{n-1}(c)+N^c\frac{n^*-1}{N}S_{n^*-1}(c)\sum_{n=n^*}^N\frac{1}{n-1}.\label{eq:piN}
\end{equation}
Letting $N\to\infty$ in the second term of \eqref{eq:piN}, we obtain
\begin{align}
    N^c\frac{n^*-1}{N}S_{n^*-1}(c)\sum_{n=n^*}^N\frac{1}{n-1}&=\left(\frac{N}{n^*-1}\right)^{c-1}(n^*-1)^cS_{n^*-1}(c)\sum_{n=n^*}^N\frac{1}{n-1} \notag \\
    &\to \e^{c-1}\cdot \frac{1}{\Gamma(1-c)}\cdot 1=\frac{\e^{c-1}}{\Gamma(1-c)}, \label{eq:lim1}
\end{align}
where the three limits follow from Proposition \ref{prop:nstar} and Lemmas \ref{lem:sumbound}, \ref{lem:Sn}.

We next evaluate the limit of the first term in \eqref{eq:piN} as $N\to\infty$. Since $c\in [0,1)$, we can choose $\alpha\in (0,1-c)$. Set $n_1=\floor{N^\alpha}$. Then $n_1\to\infty$ and $n_1/N\to 0$ as $N\to\infty$. Since $n^*/N\to 1/\e$ by Proposition \ref{prop:nstar}, we have $n_1<n^*-1$ for large enough $N$. For such $N$, write the first term in \eqref{eq:piN} as
\begin{equation*}
    N^c(c/N)\sum_{n=1}^{n^*-1}S_{n-1}(c)=N^c(c/N)\sum_{n=1}^{n_1-1}S_{n-1}(c)+N^c(c/N)\sum_{n=n_1}^{n^*-1}S_{n-1}(c).
\end{equation*}
Since $S_{n-1}(c)\le 1$ by the definition in \eqref{eq:Sn}, we can bound the first term in the right-hand side as
\begin{equation}
    N^c(c/N)\sum_{n=1}^{n_1-1}S_{n-1}(c)\le N^c(c/N)(n_1-1)\le N^c(c/N)N^\alpha= cN^{\alpha-1+c}\to 0. \label{eq:lim2}
\end{equation}
To estimate the second term, take any $\epsilon<1/\Gamma(1-c)$. Since $n\ge n_1\to\infty$ as $N\to\infty$, by Lemma \ref{lem:Sn}, we have
\begin{equation*}
    \left(\frac{1}{\Gamma(1-c)}-\epsilon\right)n^{-c}\le S_{n-1}(c)\le \left(\frac{1}{\Gamma(1-c)}+\epsilon\right)n^{-c}
\end{equation*}
for any $n\ge n_1$ for large enough $N$. Therefore
\begin{align*}
    c\left(\frac{1}{\Gamma(1-c)}-\epsilon\right)\sum_{n=n_1}^{n^*-1}(n/N)^{-c}\frac{1}{N} &\le N^c(c/N)\sum_{n=n_1}^{n^*-1}S_{n-1}(c)\\
    &\le c\left(\frac{1}{\Gamma(1-c)}+\epsilon\right)\sum_{n=n_1}^{n^*-1}(n/N)^{-c}\frac{1}{N}.
\end{align*}
Letting $N\to\infty$, using $n_1/N\to 0$ and $n^*/N\to 1/\e$, using the definition of the Riemann integral, and sending $\epsilon\to 0$, we obtain
\begin{equation}
    \lim_{n\to\infty}N^c(c/N)\sum_{n=n_1}^{n^*-1}S_{n-1}(c)=\frac{c}{\Gamma(1-c)}\int_0^{1/\e} x^{-c}\diff x=\frac{c}{\Gamma(1-c)}\frac{\e^{c-1}}{1-c}. \label{eq:lim3}
\end{equation}
Combining \eqref{eq:piN}--\eqref{eq:lim3}, it follows that
\begin{equation*}
    \lim_{N\to\infty}N^c\pi_N^*=\frac{\e^{c-1}}{\Gamma(1-c)}\left(\frac{c}{1-c}+1\right)=\frac{\e^{c-1}}{(1-c)\Gamma(1-c)}=\frac{\e^{c-1}}{\Gamma(2-c)}. \qedhere
\end{equation*}
\end{proof}

\subsection{Proof of Theorem \ref{thm:tau}}

In the full learning equilibrium, by Theorem \ref{thm:eq} the current best applicant $n$ is accepted with probability $c$ if $n<n^*$ and with probability 1 if $n\ge n^*$. Furthermore, under full learning, the probability that the applicant $n$ is the best among the first $n$ is $1/n$. Therefore on equilibrium path, the probability that applicant $n<n^*$ is accepted is
\begin{equation*}
\left[\prod_{k=1}^{n-1}\left(1-\frac{c}{k}\right)\right]\frac{c}{n}=\frac{c}{n}S_{n-1}(c),
\end{equation*}
where we have used \eqref{eq:Sn}. Similarly, the probability that applicant $n\ge n^*$ is accepted is
\begin{equation*}
\left[\prod_{k=1}^{n^*-1}\left(1-\frac{c}{k}\right)\right]\left[\prod_{k=n^*}^{n-1}\left(1-\frac{1}{k}\right)\right]\frac{1}{n}=\frac{1}{n}S_{n^*-1}(c)\frac{n^*-1}{n-1}.
\end{equation*}
Therefore the expected length of search is
\begin{equation}
\E[\tau_N^*]=c\sum_{n=1}^{n^*-1}S_{n-1}(c)+S_{n^*-1}(c)\sum_{n=n^*}^N\frac{n^*-1}{n-1}.\label{eq:EtauN}
\end{equation}
Comparing \eqref{eq:EtauN} to \eqref{eq:piN}, we obtain $\E[\tau_N^*]=N\pi_N^*$, and the limit \eqref{eq:tau} holds. \hfill \qedsymbol

\printbibliography

\end{document}